\documentclass{article}
\usepackage[fleqn]{amsmath}
\usepackage{spconf}
\usepackage{tikz}
\usepackage{dsfont}
\usepackage{algorithm}
\usepackage{algpseudocode}
\usepackage{amsthm}
\usepackage{appendix}
\usepackage{multirow}
\usepackage{xpatch}

\expandafter\patchcmd\csname\string\algorithmic\endcsname{\labelwidth 0.5em}{\labelwidth0pt\labelsep0pt}{}{}

\newtheorem{lemma}{Lemma}
\newtheorem*{theorem}{Theorem}

\def\n{{(n)}}
\def\m{{(m)}}
\def\a{{(a)}}
\def\E{\mathds{E}}
\def\S{\mathcal{S}}
\def\var{\operatorname{var}}
\def\ESS{\text{ESS}}

\title{Particle filter with rejection control and unbiased estimator of the marginal likelihood}
\name{
    Jan Kudlicka$^\star$ \qquad Lawrence M.~Murray$^\dagger$ \qquad Thomas B.~Sch\"on$^\star$ \qquad Fredrik Lindsten$^\ddagger$
    \thanks{This research was supported by the Swedish Foundation for Strategic Research via the project ASSEMBLE (contract number: RIT15-0012) and by the Swedish Research Council grants 2013-4853 and 2017-03807.}
}
\address{
    $^\star$ Department of Information Technology, Uppsala University, Uppsala, Sweden \\
    $^\dagger$ Uber AI, San Francisco, CA, USA \\
    $^\ddagger$ Division of Statistics and Machine Learning, Link\"oping University, Link\"oping, Sweden
}

\begin{document}

\onecolumn

\section*{IEEE Copyright Notice}

\copyright\ 2019 IEEE. Personal use of this material is permitted. Permission from IEEE must be obtained for all
other uses, in any current or future media, including reprinting/republishing this material for advertising
or promotional purposes, creating new collective works, for resale or redistribution to servers or lists, or
reuse of any copyrighted component of this work in other works.

\subsection*{Accepted to be published in}

Proceedings of ICASSP 2020, IEEE International Conference on Acoustics, Speech and Signal Processing, Barcelona, Spain, May 4--8, 2020.

\vspace{1pc}\noindent
Please cite this version:

\begin{verbatim}
@inproceedings{,
  title={Particle filter with rejection control and unbiased estimator of
         the marginal likelihood},
  author={Kudlicka, Jan and Murray, Lawrence M. and Sch\"on, Thomas B. and
          Lindsten, Fredrik},
  booktitle={{IEEE} International Conference on Acoustics, Speech and Signal Processing,
             {ICASSP} 2020, Barcelona, Spain, May 4-8, 2020},
  publisher={{IEEE}},
  year={2020}
}
\end{verbatim}

\noindent
Note that the accepted version does not include the appendices that are included in the extended version (this document).

\clearpage
\twocolumn

\maketitle

\begin{abstract}
We consider the combined use of resampling and partial rejection control in sequential Monte Carlo methods, also known as particle filters. While the variance reducing properties of rejection control are known, there has not been (to the best of our knowledge) any work on unbiased estimation of the marginal likelihood (also known as the model evidence or the normalizing constant) in this type of particle filter. Being able to estimate the marginal likelihood without bias is highly relevant for model comparison, computation of interpretable and reliable confidence intervals, and in \emph{exact approximation} methods, such as particle Markov chain Monte Carlo. In the paper we present a particle filter with rejection control that enables unbiased estimation of the marginal likelihood.
\end{abstract}
\begin{keywords}
Particle filters, sequential Monte Carlo (SMC), partial rejection control, unbiased estimate of the marginal likelihood
\end{keywords}

\section{INTRODUCTION}

Rejuvenation of particles in methods based on sequential importance sampling is a crucial step to avoid the weight degeneracy problem. Sequential Monte Carlo (SMC) methods typically use \emph{resampling}, but there are alternative methods available. \emph{Rejection control}, proposed by Liu et al. \cite{liu1998rejection} solves the degeneracy problem by checking the weights of particles (or \emph{streams} in their terminology) at given checkpoints, and comparing them to given thresholds. Particles with weights below a certain checkpoint threshold are probabilistically discarded and replaced by new particles that are restarted from the beginning. Discarding particles that have passed through all previous checkpoints is quite disadvantageous. Liu proposed a modified version of the algorithm in \cite{liu2008monte}, called \emph{partial rejection control}, where a set of particles is propagated in parallel between the checkpoints, and each rejected particle gets replaced by a sample drawn from the particle set at the previous checkpoint (quite similar to resampling), and propagated forward, rather than restarting from the beginning.

Peters et al. \cite{peters2012sequential} combined partial rejection control and resampling: the resampling, propagation and weighting steps are the same as in standard SMC methods, but an additional step is performed after the weighting step. Here, the weight of each particle is compared to a threshold and if it falls below this threshold, the particle is probabilistically rejected, and the resampling, propagation and resampling steps are repeated. This procedure is repeated until the particle gets accepted. Peters et al. also adapted the algorithm to be used in an approximate Bayesian computation (ABC) setting.

We consider models with likelihoods and our contribution is a non-trivial modification of the particle filter with rejection control, allowing us to define an unbiased estimator of the marginal likelihood. This modification is very simple to implement: it requires an additional particle and counting the number of propagations. The unbiasedness of the marginal likelihood estimator opens for using particle filters with rejection control in exact approximate methods such as particle marginal Metropolis-Hastings (PMMH, \cite{AndrieuDH:2010}), model comparison, and computation of interpretable and reliable confidence intervals.

\section{BACKGROUND}

\subsection{State space model}

\emph{State space models} are frequently used to model dynamical systems where the state evolution exhibits the Markov property (i.e., the state at time $t$ only depends on the state at time $t-1$ but not on the state at any earlier time). Further, the state is not observed directly, but rather via measurements depending (stochastically) only on the state at the same time.

Let $x_t$ denote the state at time $t$ and $y_t$ the corresponding measurement. The state space model can be represented using probability distributions:
\begin{equation*}
    x_0 \sim \mu_0(\cdot), \hspace{1.8pc} x_t \sim f_{t}(\cdot | x_{t-1}), \hspace{1.8pc} y_t \sim g_{t}(\cdot | x_t).
\end{equation*}

The inference goal is to estimate posterior distributions of (a subset of) the state variables given a set of measurements, and to estimate the expected value of test functions with respect to these distributions. We are usually interested in the joint filtering distribution $p(x_{0:t}|y_{1:t})$ and the filtering distribution $p(x_t|y_{1:t})$, where $x_{0:t}$ denotes the sequence of all states until time $t$, i.e.~$x_0, x_1, \dots, x_t$, and similarly for $y_{1:t}$.

\subsection{Particle filters}

Particle filters are sampling-based methods that construct sets $\mathcal{S}_t$ of $N$ weighted samples (particles) to estimate the filtering distribution $p(x_t|y_{1:t})$ for each time $t$. The baseline \emph{bootstrap particle filter} creates the initial set $\mathcal{S}_0$ by drawing $N$ samples from $\mu_0$ and setting their (unnormalized) weights to 1, i.e.,
\begin{equation*}
    \mathcal{S}_0 = \left\{(x^\n_0, w^\n_0) \ \Big|\ x^\n_0 \sim \mu_0, w^\n_0 = 1\right\}_{n=1}^N.
\end{equation*}
The set $\mathcal{S}_t$ at time $t$ is constructed from the previous set of particles $\mathcal{S}_{t-1}$ by repeatedly ($N$ times) choosing a particle from $\mathcal{S}_{t-1}$ with probabilities proportional to their weights (this step is called \emph{resampling}), drawing a new sample $x$ from $f_t(\cdot | x^\a_{t-1})$ where $a$ is the index of the chosen particle (\emph{propagation}), and setting its weight to $g_t(y_t | x)$ (\emph{weighting}):
\begin{align*}
    \mathcal{S}_t =& \left\{(x^\n_t, w^\n_t) \ \Big|\ a_n \sim \mathcal{C}(\{w^\m_{t-1}\}_{m=1}^N) \right.\\
                   & \hspace{5.3pc} \left. x^\n_t \sim f_t(\cdot|x^{(a_n)}_{t-1}), \right.  \\
                   & \hspace{5.3pc} \left. w^\n_t = g_t(y_t|x^\n_t) \right\}_{n=1}^N.
\end{align*}
Here, $\mathcal{C}$ denotes the categorical distribution with the unnormalized event probabilities specified as the parameter. The crucial element of a particle filter is the resampling step that avoids the weight degeneracy problem of sequential importance sampling. The pseudocode for the bootstrap particle filter is listed as Algorithm~\ref{alg:bpf}.

\begin{algorithm}[t]
\caption{Bootstrap particle filter (BPF)}
\label{alg:bpf}
\begin{algorithmic}
\State $\widehat Z_\text{BPF} \gets 1$
\For {$n=1\textbf{\ to\ }N$} \Comment{{\small Initialize}}
    \State $x^\n_0 \sim \mu_0$, $w^\n_0 \gets 1$
\EndFor
\For {$t=1\textbf{\ to\ }T$}
    \For {$n=1\textbf{\ to\ }N$}
        \State $a^\n \sim \mathcal{C}(\{w^\m_{t-1}\}_{m=1}^N)$ \Comment{{\small Resample}}
        \State $x^\n_t \sim f_t(\cdot|x^{(a^\n)}_{t-1})$ \Comment{{\small Propagate}}
        \State $w_t^\n \gets g_t(y_t|x^\n_t)$ \Comment{{\small Weight}}
    \EndFor
    \State $\widehat Z_\text{BPF} \gets \widehat Z_\text{BPF} \sum_{n=1}^N w_t^\n / N$
\EndFor
\end{algorithmic}

\end{algorithm}

The unbiased estimator $\widehat Z_\text{BPF}$ of the marginal likelihood $p(y_{1:T})$ (unbiased in the sense that $\E[\widehat Z_\text{BPF}] = p(y_{1:T})$) is given by \cite{moral2004feynman}:
\begin{equation*}
\widehat Z_\text{BPF} = \prod_{t=1}^T \frac{1}{N} \sum_{n=1}^N w_t^\n.
\end{equation*}

Particle filters are a family of different variants of this algorithm. These variants use different proposal distributions to choose the initial set of samples, to propagate or to resample particles and use appropriate changes in the calculation of the importance weights w.r.t.~the filtering distributions. In order to reduce the variance of estimators, some methods do not resample at every time step, but rather only when a summary statistic of weights crosses a given threshold, e.g., when the effective sample size (ESS) falls below $\nu N$, where $\nu \in [0, 1]$ is a tuning parameter.

In probabilistic programming, the state space model can be used to model program execution and particle filters are used as one of the general probabilistic programming inference methods. Examples of probabilistic programming languages that use particle filters and SMC for inference include Anglican \cite{tolpin2016design}, Biips \cite{todeschini2014biips}, Birch \cite{murray2018automated}, Figaro \cite{pfeffer2016practical}, LibBi \cite{murray2015bayesian}, Venture \cite{mansinghka2014venture}, WebPPL \cite{goodman2014design} and Turing \cite{ge2018turing}.

\section{PARTICLE FILTER WITH REJECTION CONTROL}

In certain models, such as models with jump processes or rare-event processes, or when the measurements contain outliers, the weights of many particles after propagation and weighting might be rather low or even zero. This decreases the ESS and thus increases the variance of the estimators of interest.

Below we present the \emph{particle filter with rejection control} (PF-RC) that ensures that the weights of all particles in the particle set $\S_t$ are greater than or equal to a chosen threshold, denoted by $c_t > 0$. The process of drawing new particles in PF-RC is almost identical to the process for the bootstrap particle filter described in the previous section, with one additional step that we will refer to as the \emph{acceptance} step, described in the next paragraph.

Let $(x'^\n_t, w'^\n_t)$ denote the particle after the resampling, propagation and weighting steps. The particle is accepted (and added to $\S_t$) with probability $\min(1, w'^\n_t/c_t)$. If accepted, the particle weight is lifted to $\max(w'^\n_t, c_t)$. If the particle is rejected, the resampling, propagation, weighting and acceptance steps are repeated until acceptance. The following table summarizes the acceptance step:
\begin{center}
\begin{tabular}{@{}lcll@{}}
Condition & Acc.~prob. & If accepted & If rejected \\
\hline
$w'^\n_t \ge c_t$ & 1 & $w^\n_t \gets w'^\n_t$ & --- \\
& & $x^\n_t \gets x'^\n_t$ & \\
$w'^\n_t < c_t$ & $w'^\n_t/c_t$ & $w^\n_t \gets c_t$ & Sample new \\
& & $x^\n_t \gets x'^\n_t$ & $x'^\n_t$ \\
\end{tabular}
\end{center}
An important difference compared to the bootstrap particle filter is that we also use the same procedure (i.e., the resampling, propagation and acceptance steps) to sample one additional particle. This particle is \emph{not} added to the particle set and its weight is not used either, but the number of propagations until its acceptance is relevant to the estimation of the marginal likelihood $p(y_{1:T})$.

Let $P_t$ denote the total number of propagation steps performed in order to construct $S_t$ as well as the additional particle. By the total number we mean that the propagation steps for both rejected and accepted particles are counted. The estimate $\widehat Z$ of the marginal likelihood $p(y_{1:T})$ is given by
\begin{equation*}
    \widehat Z = \prod_{t=1}^T \frac{\sum_{n=1}^N w^\n_t}{P_t - 1}.
\end{equation*}

\begin{theorem}
The marginal likelihood estimator $\widehat Z$ is unbiased in sense that $\E[\widehat Z] = p(y_{1:T})$.
\end{theorem}
\begin{proof}
See Appendix A.
\end{proof}

\begin{algorithm}[t]
\caption{Particle filter with rejection control (PF-RC)}
\label{alg:pf-rc}
\begin{algorithmic}
\color{gray}
\State $\widehat Z \gets 1$
\For {$n=1\textbf{\ to\ }N$} \Comment{{\small Initialize}}
    \State $x^\n_0 \sim \mu_0$, $w^\n_0 \gets 1$
\EndFor
\For {$t=1\textbf{\ to\ }T$}
    \color{black}\State $P_t \gets 0$\color{gray}
    \For {$n=1\textbf{\ to\ }N$}
        \color{black}\Repeat\color{gray}
            \State $a^\n \sim \mathcal{C}(\{w^\m_{t-1}\}_{m=1}^N)$ \Comment{{\small Resample}}
            \State $x^\n_t \sim f_t(\cdot|x^{(a^\n)}_{t-1})$ \Comment{{\small Propagate}}
            \State $w_t^\n \gets g_t(y_t|x^\n_t)$ \Comment{{\small Weight}}
            {\color{black}\State $P_t \gets P_t + 1$}
            {\color{black}\State $\alpha \sim \operatorname{Bernoulli}(\min(1, w^\n_t  / c_t))$}
        \color{black}\Until{$\alpha$} \Comment{{\small Accept / reject}}
        \color{black}\State $w^\n_t \gets \max(w^\n_t, c_t)$ \Comment{{\small Update the weight}}
    \color{gray}\EndFor
    \color{black}
    \Repeat \Comment{{\small Additional particle}}
        \State $a' \sim \mathcal{C}(\{w^\m_{t-1}\}_{m=1}^N)$
        \State $x' \sim f_t(\cdot|x^{(a')}_{t-1})$
        \State $w' \gets g_t(y_t|x')$
        \State $P_t \gets P_t + 1$
        \State $\alpha \sim \operatorname{Bernoulli}(\min(1, w' / c_t))$
    \Until{$\alpha$}
    \State $\widehat Z \gets \widehat Z (\sum_{n=1}^N w^\n_t) / (P_t - 1)$
\EndFor
\end{algorithmic}

\end{algorithm}

The pseudocode for the PF-RC is listed as Algorithm~\ref{alg:pf-rc}. The gray color marks the part that is the same as in the bootstrap particle filter.

The question remains of how to choose the thresholds $\{c_t\}$. One option is to use some prior knowledge (that can be obtained by pilot runs of a particle filter) and choose fixed thresholds. Liu et al. \cite{liu1998rejection} mention several options for determining the thresholds dynamically after propagating and weighting all particles for the first time at each time step, using a certain quantile of these weights or a weighted average of the minimum, average and maximum weight, i.e., $c_t = p_1 \min w'_t + p_2 \bar{w}'_t + p_3 \max w'_t$, where all $p_i > 0$ and $p_1 + p_2 + p_3 = 1$. In general, setting the thresholds dynamically in each run breaks the unbiasedness of the marginal likelihood estimator (as demonstrated by an example in Appendix B).

Note that the alive particle filter \cite{kudlicka2019probabilistic} can be obtained as a limiting case of the particle filter with rejection control when all $c_t \rightarrow 0$. Instead of $\alpha \sim \operatorname{Bernoulli}(\min(1, w/c_t))$ we need to use $\alpha \gets \text{true}$ if $w > 0$ and false otherwise, but the rest of the algorithm remains the same.

\section{EXPERIMENTS}

\subsection{Linear Gaussian state space model with outliers}

\begin{table}[t]
    \centering
    \caption{Comparison of the filters for the model with outliers. See also the description in the text.}
    \label{tab:lgss_summary}
    \vspace{0.5pc}
    \begin{tabular}{@{}lr|r|r|r|r|r|r@{}}
        \multicolumn{2}{@{}r|}{\multirow{2}{*}{$c$}} & \multirow{2}{*}{$N$} & \multirow{2}{*}{$\rho$} & \multirow{2}{*}{$\ESS$} & $\ESS$ & $\var$ & $\rho\var$ \\
        \multicolumn{2}{@{}r|}{} & & & & $/\rho$ & $\log\widehat Z$ & $\log\widehat Z$ \\
        \hline
        \multicolumn{2}{@{}l|}{BPF} & 1024 & 1.00 & 101.6 & 101.6 & 2.18 & 2.18 \\
        \hline
        \multirow{7}{*}{\rotatebox{90}{PF-RC}} 
        & $10^{-14}$ & \multirow{7}{*}{1024} & 1.06 & 180.1 & 169.8 & 1.13 & 1.20 \\
        & $10^{-13}$ & & 1.08 & 285.8 & 264.0 & 1.08 & 1.17 \\
        & $10^{-12}$ & & 1.12 & 386.1 & 346.1 & 1.02 & 1.14 \\
        & $10^{-11}$ & & 1.17 & 460.2 & \bf 394.8 & 0.90 & 1.05 \\
        & $10^{-10}$ & & 1.25 & 471.0 & 377.9 & 0.87 & 1.08 \\
        & $10^{-9}$ & & 1.38 & 491.9 & 356.3 & 0.76 & \bf 1.04 \\
        & $10^{-8}$ & & 1.62 & \bf 568.0 & 350.4 & \bf 0.65 & 1.06 \\
        \hline
        \multicolumn{2}{@{}l|}{BPF} & 1200 & 1.17 & 185.6 & 158.3 & 1.91 & 2.24 \\
    \end{tabular}
\end{table}

The particle filter with rejection control may be useful in situations where measurements include outliers. To demonstrate this we considered the following linear Gaussian state space model:
\begin{align*}
    x_0 \sim \mathcal{N}(0, 0.25), \hspace{2pc} x_t &\sim \mathcal{N}(0.8 x_{t-1}, 0.25), \\
    y_t &\sim \mathcal{N}(x_t, 0.1).
\end{align*}
We simulated a set of measurements with outliers by replacing the measurement equation with $y_t \sim 0.9\,\mathcal{N}(x_t, 0.1) + 0.1\,\mathcal{N}(0, 1)$, and used these measurements in a set of experiment with both the BPF and a set of PF-RC with the thresholds at all checkpoints equal to a given value, i.e. $c_t = c$, where $c \in \{10^{-14}, 10^{-13}, \dots, 10^{-8}\}$. We ran each filter $M=1000$ times using $N=1024$ particles, collected a set of the estimates $\{\widehat Z_m\}_{m=1}^M$ of the marginal likelihood, and calculated several summary statistics, presented in Table~\ref{tab:lgss_summary}. Here, the effective sample size ESS means $(\sum_{m=1}^M \widehat Z_m)^2 / \sum_{m=1}^M \widehat Z_m^2$, and $\rho$ denotes the average number of propagations relatively to the number of propagation in the bootstrap particle filter with the same number of particles.

The filter that maximized $\ESS/\rho$ used 1.17 times more propagations than the baseline BPF, so we also repeated the experiment using a BPF with 1200 particles to match the number of propagations; the results are shown in the last row.

\subsection{Object tracking}

We implemented the particle filter with rejection control as an inference method in the probabilistic programming language Birch \cite{murray2018automated}. We used the multiple object tracking model from \cite{murray2018automated} but restricted it to two objects that both appear at the initial time within the target area for computational purposes.

We used the program to simulate the tracks and measurements for 50 time steps. We ran experiments using 50-th, 60-th, 70-th, 80-th, 90-th, 95-th and 99-th percentiles as the thresholds, but in order to keep the marginal likelihood estimates unbiased, we first ran the program once for each of the percentiles (using 32768 particles) and saved the threshold values. We then ran the program $M=100$ times for each of the percentiles, using the saved values as fixed thresholds and with only $N=4096$ particles. In case where the threshold was 0, we fell back to accepting all particles. We compared the marginal likelihood estimates with the estimates obtained by running the program with the BPF inference. The results are summarized in Table~\ref{tab:mot_summary} and Fig.~\ref{fig:mot_logz}.

We also repeated the experiments using a bootstrap particle filter with 32650 particles (to match the number of propagations in the filter with the 99-th percentile), the results are presented in the last row of the table.

\begin{table}[t]
    \centering
    \caption{Comparison of the filters for the object tracking problem.}
    \label{tab:mot_summary}
    \vspace{0.5pc}
    \begin{tabular}{@{}lr|r|r|r|r|r|r@{}}
        \multicolumn{2}{@{}r|}{Per-} & \multirow{2}{*}{$N$} & \multirow{2}{*}{$\rho$} & \multirow{2}{*}{$\ESS$} & $\ESS$ & $\var$ & $\rho\var$ \\
        \multicolumn{2}{@{}r|}{centile} & & & & $/\rho$ & $\log\widehat Z$ & $\log\widehat Z$ \\
        \hline
        \multicolumn{2}{@{}l|}{BPF} & 4096 & 1.00 & 8.5 & \bf 8.5 & 482.00 & 482.00 \\
        \hline
        \multirow{6}{*}{\rotatebox{90}{PF-RC}}
        & 50 & \multirow{7}{*}{4096} & 2.68 & 4.3 & 1.6 & 94.75 & 253.68 \\
        & 60 & & 2.13 & 6.8 & 3.2 & 44.58 & 94.83 \\
        & 70 & & 1.99 & 8.7 & 4.4 & 35.22 & 70.18 \\
        & 80 & & 2.35 & 10.1 & 4.3 & 32.30 & 75.95 \\
        & 90 & & 3.39 & 13.7 & 4.0 & 14.63 & 49.54 \\
        & 95 & & 4.78 & 32.0 & 6.7 & 3.95 & 18.87 \\
        & 99 & & 7.97 & \bf 44.6 & 5.6 & \bf 1.08 & \bf 8.65 \\
        \hline
        \multicolumn{2}{@{}l|}{BPF} & 32650 & 7.97 & 10.7 & 1.3 & 14.39 & 114.67 \\
    \end{tabular}
\end{table}

\begin{figure}[t]
    \centering
    \resizebox{\columnwidth}{!}{\input{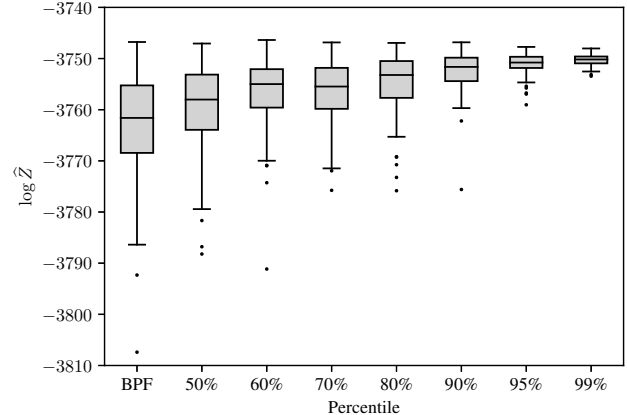}}
    \caption{Box plot of $\log\widehat Z$ for the object tracking problem.}
    \label{fig:mot_logz}
\end{figure}

\section{DISCUSSION AND CONCLUSION}

In this paper we presented a particle filter with rejection control (PF-RC) that enables unbiased estimation of the marginal likelihood. We briefly mentioned several situations where the unbiasedness is important. We also showed a couple of examples that demonstrated the potential of the method. As we saw, the PF-RC outperformed (in terms of ESS and $\var\log\widehat Z$) the bootstrap particle filter (BPF) even when the latter used more particles and matched the total number of propagations. This is due to the fact that the number of propagations varies between the time steps in PF-RC, and more propagations are used ``where it is needed'', while BPF uses the same number of propagations at each time step. The PF-RC also needs to use less memory compared to the BPF with the matched number of propagations, which might be an important advantage in problems requiring many particles. On the other hand, it might not always be clear how to determine the thresholds. In our future work we wish to look into this question, especially in the context of using PF-RC in exact approximate methods such as particle marginal Metropolis-Hastings method.

\bibliographystyle{IEEEbib}
\bibliography{refs}

\begin{thebibliography}{10}

\bibitem{liu1998rejection}
Jun~S Liu, Rong Chen, and Wing~Hung Wong,
\newblock ``Rejection control and sequential importance sampling,''
\newblock {\em Journal of the American Statistical Association}, vol. 93, no.
  443, pp. 1022--1031, 1998.

\bibitem{liu2008monte}
Jun~S Liu,
\newblock {\em Monte Carlo strategies in scientific computing},
\newblock Springer Science \& Business Media, 2008.

\bibitem{peters2012sequential}
Gareth~W Peters, Yanan Fan, and Scott~A Sisson,
\newblock ``On sequential {M}onte {C}arlo, partial rejection control and
  approximate {B}ayesian computation,''
\newblock {\em Statistics and Computing}, vol. 22, no. 6, pp. 1209--1222, 2012.

\bibitem{AndrieuDH:2010}
C.~Andrieu, A.~Doucet, and R.~Holenstein,
\newblock ``Particle {M}arkov chain {M}onte {C}arlo methods,''
\newblock {\em Journal of the Royal Statistical Society: Series B}, vol. 72,
  no. 3, pp. 269--342, 2010.

\bibitem{moral2004feynman}
Pierre Del~Moral,
\newblock {\em {F}eynman-{K}ac Formulae: Genealogical and Interacting Particle
  Systems with Applications},
\newblock Probability and Its Applications. Springer New York, 2004.

\bibitem{tolpin2016design}
David Tolpin, Jan~Willem van~de Meent, Hongseok Yang, and Frank Wood,
\newblock ``Design and implementation of probabilistic programming language
  {A}nglican,''
\newblock {\em arXiv preprint arXiv:1608.05263}, 2016.

\bibitem{todeschini2014biips}
Adrien Todeschini, Fran{\c{c}}ois Caron, Marc Fuentes, Pierrick Legrand, and
  Pierre Del~Moral,
\newblock ``{B}iips: {S}oftware for {B}ayesian inference with interacting
  particle systems,''
\newblock {\em arXiv preprint arXiv:1412.3779}, 2014.

\bibitem{murray2018automated}
Lawrence~M Murray and Thomas~B Sch\"on,
\newblock ``Automated learning with a probabilistic programming language:
  Birch,''
\newblock {\em Annual Reviews in Control}, vol. 46, pp. 29--43, 2018.

\bibitem{pfeffer2016practical}
Avi Pfeffer,
\newblock {\em Practical probabilistic programming},
\newblock Manning Publications, 2016.

\bibitem{murray2015bayesian}
Lawrence~M Murray,
\newblock ``Bayesian state-space modelling on high-performance hardware using
  {LibBi},''
\newblock {\em Journal of Statistical Software}, vol. 67, no. 10, pp. 1--36,
  2015.

\bibitem{mansinghka2014venture}
Vikash Mansinghka, Daniel Selsam, and Yura Perov,
\newblock ``Venture: a higher-order probabilistic programming platform with
  programmable inference,''
\newblock {\em arXiv preprint arXiv:1404.0099}, 2014.

\bibitem{goodman2014design}
Noah~D Goodman and Andreas Stuhlm\"{u}ller,
\newblock ``The design and implementation of probabilistic programming
  languages,'' http://dippl.org, 2014,
\newblock Accessed: 2019-10-21.

\bibitem{ge2018turing}
Hong Ge, Kai Xu, and Zoubin Ghahramani,
\newblock ``Turing: a language for flexible probabilistic inference,''
\newblock in {\em International Conference on Artificial Intelligence and
  Statistics, {AISTATS} 2018, 9-11 April 2018, Playa Blanca, Lanzarote, Canary
  Islands, Spain}, 2018, pp. 1682--1690.

\bibitem{kudlicka2019probabilistic}
Jan Kudlicka, Lawrence~M Murray, Fredrik Ronquist, and Thomas~B Sch\"on,
\newblock ``Probabilistic programming for birth-death models of evolution using
  an alive particle filter with delayed sampling,''
\newblock in {\em Conference on Uncertainty in Artificial Intelligence (UAI)},
  2019.

\end{thebibliography}

\begin{appendices}
\onecolumn

\section{PROOF OF UNBIASEDNESS OF THE MARGINAL LIKELIHOOD ESTIMATOR}

The proof follows the proof of unbiasedness of the marginal likelihood estimator for the alive particle filter given in \cite{kudlicka2019probabilistic}.

\begin{lemma}
\label{lemma:1}
\begin{equation*}
\E\left[\frac{\sum_{n=1}^N w_t^\n}{P_t - 1} \middle| \S_{t-1} \right] = \sum_{n=1}^N \frac{w_{t-1}^\n}{\sum_{m=1}^N w_{t-1}^\m} p\left(y_t \middle| x_{t-1}^\n\right).
\end{equation*}
\end{lemma}

\begin{proof}
For brevity we omit conditioning on $\S_{t-1}$ in the notation. A candidate sample $x'$ is constructed by drawing a sample from $\S_{t-1}$ with the probabilities proportional to the weights $\{w_{t-1}^\n\}$ and propagating it forward to time $t$, i.e.
\begin{equation*}
x' \sim \sum_{n=1}^N \frac{w_{t-1}^\n}{\sum_{m=1}^N w_{t-1}^\m} f_t\left(x' \middle| x_{t-1}^\n\right). 
\end{equation*}
The candidate sample $x'$ is accepted with probability $\min(1, g_t(y_t|x')/c_t)$. If the sample is rejected, a new candidate sample is drawn from the above-mentioned distribution. The acceptance probability $p_{A_t}$ is given by
\begin{equation*} p_{A_t} = \int \min\left(1, \frac{g_t(y_t|x')}{c_t}\right) \sum_{n=1}^N \frac{w_{t-1}^\n} {\sum_{m=1}^N w_{t-1}^\m} f_t\left(x' \middle| x_{t-1}^\n\right) dx'.
\end{equation*}
Accepted samples are distributed according to the following distribution:
\begin{equation*}
x_t \sim \frac{1}{p_{A_t}} \min\left(1, \frac{g_t(y_t|x_t)}{c_t}\right) \sum_{n=1}^N \frac{w_{t-1}^\n}{\sum_{m=1}^N w_{t-1}^\m} f_t\left(x_t \middle| x_{t-1}^\n\right)
\end{equation*}
and the expected value of the weight $w_t = \max(g_t(y_t|x_t), c_t)$ of an accepted sample is given by
\begin{equation*}
\E[w_t] = \int \max\left(g_t(y_t|x_t), c_t\right) \frac{1}{p_{A_t}} \min\left(1, \frac{g_t(y_t|x_t)}{c_t}\right) \sum_{n=1}^N \frac{w_{t-1}^\n}{\sum_{m=1}^N w_{t-1}^\m} f_t\left(x_t \middle| x_{t-1}^\n\right) dx_t.
\end{equation*}
Note that $\max\left(g_t(y_t|x_t), c_t\right) \times \min(1, g_t(y_t|x_t)/c_t) = g_t(y_t|x_t)$. To prove that, consider two cases: if $g_t(y_t|x_t) \geq c_t$, the result of the multiplication is $g_t(y_t|x_t) \times 1 = g_t(y_t|x_t)$; if $g_t(y_t|x_t) < c_t$, the result is $c_t \times g_t(y_t|x_t) / c_t = g_t(y_t|x_t)$. (This also gives an intuition about why the weight gets lifted to $c_t$ in the case of acceptance with $w_t < c_t$.) Using this we have that:
\begin{align*}
\E[w_t] &= \int g_t(y_t|x_t) \frac{1}{p_{A_t}} \sum_{n=1}^N \frac{w_{t-1}^\n}{\sum_{m=1}^N w_{t-1}^\m} f_t\left(x_t \middle| x_{t-1}^\n\right)  dx_t = \frac{1}{p_{A_t}} \sum_{n=1}^N \frac{w_{t-1}^\n}{\sum_{m=1}^N w_{t-1}^\m} \int f_t\left(x_t \middle| x_{t-1}^\n\right) g_t(y_t|x_t) dx_t \\
&= \frac{1}{p_{A_t}} \sum_{n=1}^N \frac{w_{t-1}^\n}{\sum_{m=1}^N w_{t-1}^\m}\ p\left(y_t \middle| x_{t-1}^\n\right).
\end{align*}

The number of propagations $P_t$ is a random variable distributed according to the negative binomial distribution with the number of successes $N+1$ and the probability of success $p_{A_t}$:
\begin{equation*}
P(P_t = D) = \binom{D-1}{(N+1)-1} p_{A_t}^{N+1} (1-p_{A_t})^{D-(N+1)}.
\end{equation*}

The expected value of $\sum_{n=1}^N  w_t^\n / (P_t - 1)$ is given by
\begin{align*}
\E\left[\frac{\sum_{n=1}^N w_t^\n}{P_t - 1}\right] &= \sum_{D=N+1}^\infty \frac{N \E[w_t]}{D-1} \binom{D-1}{N} p_{A_t}^{N+1} (1-p_{A_t})^{D-(N+1)} \\
&= N \E[w_t] \sum_{D=N+1}^\infty \frac{1}{D-1} \binom{D-1}{N} p_{A_t}^{N+1} (1-p_{A_t})^{D-(N+1)} = N \E[w_t] \frac{p_{A_t}}{N} \\
&= \sum_{n=1}^N \frac{w_{t-1}^\n}{\sum_{m=1}^N w_{t-1}^\m}\ p\left(y_t \middle| x_{t-1}^\n\right).
\end{align*}
\end{proof}

The rest of the proof is identical to the proof in \cite{kudlicka2019probabilistic}, which we include below for completeness.

\begin{lemma}
\label{lemma:2}
\begin{equation*}
\E\left[\frac{\sum_{n=1}^N w_t^\n p\left(y_{t+1:t'} \middle| x_t^\n\right)}{P_t - 1} \middle| \mathcal{S}_{t-1} \right] = \sum_{n=1}^N \frac{w_{t-1}^\n}{\sum_{m=1}^N w_{t-1}^\m} p\left(y_{t:t'} \middle| x_{t-1}^\n\right).
\end{equation*}
\end{lemma}

\begin{proof}
Similar to the proof of Lemma \ref{lemma:1} we have that
\begin{align*}
\E[w_t p(y_{t+1:t'}|x_t)] &= \int \frac{1}{p_{A_t}} \sum_{n=1}^N \frac{w_{t-1}^\n}{\sum_{m=1}^N w_{t-1}^\m} f_t\left(x_t \middle| x_{t-1}^\n\right) g_t(y_t|x_t) p(y_{t+1:t'}|x_t) dx_t \\ 
&= \frac{1}{p_{A_t}} \sum_{n=1}^N \frac{w_{t-1}^\n}{\sum_{m=1}^N w_{t-1}^\m} \int f_t\left(x_t \middle| x_{t-1}^\n\right) g_t(y_t|x_t) p(y_{t+1:t'}|x_t) dx_t \\
&= \frac{1}{p_{A_t}} \sum_{n=1}^N \frac{w_{t-1}^\n}{\sum_{m=1}^N w_{t-1}^\m}\ p\left(y_{t:t'} \middle| x_{t-1}^\n\right).
\end{align*}
Using this result we have that
\begin{align*}
\E\left[\frac{\sum_{n=1}^N w_t^\n p\left(y_{t+1:t'} \middle| x_t^\n\right)}{P_t - 1}\right] &= \sum_{D=N+1}^\infty \frac{N \E[w_t p(y_{t+1:t'}|x_t)]}{D-1} \binom{D-1}{N} p_{A_t}^{N+1} (1-p_{A_t})^{D-(N+1)} \\
&= N \E[w_t p(y_{t+1:t'}|x_t)] \sum_{D=N+1}^\infty \frac{1}{D-1} \binom{D-1}{N} p_{A_t}^{N+1} (1-p_{A_t})^{D-(N+1)} \\
&= N \E[w_t p(y_{t+1:t'}|x_t)] \frac{p_{A_t}}{N} = \sum_{n=1}^N \frac{w_{t-1}^\n}{\sum_{m=1}^N w_{t-1}^\m}\ p\left(y_{t:t'} \middle| x_{t-1}^\n\right).
\end{align*}

\end{proof}

\begin{lemma}
\label{lemma:3}
\begin{equation*}
\E\left[\prod_{t'=t-h}^t \frac{\sum_{n=1}^N w_{t'}^\n}{P_{t'} - 1} \middle| \mathcal{S}_{t-h-1} \right] = \sum_{n=1}^N \frac{w_{t-h-1}^\n}{\sum_{m=1}^N w_{t-h-1}^\m} p\left(y_{t-h:t} \middle| x_{t-h-1}^\n\right).
\end{equation*}
\end{lemma}

\begin{proof}
By induction. The base step for $h=0$ was proved in Lemma \ref{lemma:1}. In the induction step, let us assume that the equality holds for $h$ and prove it for $h+1$:
\begin{align*}
\E\left[\prod_{t'=t-h-1}^t \frac{\sum_{n=1}^N w_{t'}^\n}{P_t' - 1} \middle| \mathcal{S}_{t-h-2} \right]
&= \E\left[\E\left[\prod_{t'=t-h}^t \frac{\sum_{n=1}^N w_{t'}^\n}{P_t' - 1} \middle| \mathcal{S}_{t-h-1} \right] \frac{\sum_{n=1}^N w_{t-h-1}^\n}{P_{t-h-1} - 1} \middle| \mathcal{S}_{t-h-2} \right]\\
& \text{(using the induction assumption)}\\
&= \E\left[   \sum_{n=1}^N \frac{w_{t-h-1}^\n}{\sum_{m=1}^N w_{t-h-1}^\m} p\left(y_{t-h:t} \middle| x_{t-h-1}^\n\right) \frac{\sum_{n=1}^N w_{t-h-1}^\n}{P_{t-h-1} - 1} \middle| \mathcal{S}_{t-h-2} \right] \\
&= \E\left[   \sum_{n=1}^N \frac{w_{t-h-1}^\n}{P_{t-h-1} - 1} p\left(y_{t-h:t} \middle| x_{t-h-1}^\n\right) \middle| \mathcal{S}_{t-h-2} \right] \\
& \text{(using Lemma \ref{lemma:2})}\\
&= \sum_{n=1}^N \frac{w_{t-h-2}^\n}{\sum_{m=1}^N w_{t-h-2}^\m} p\left(y_{t-h-1:t} \middle| x_{t-h-2}^\n\right).
\end{align*}
\end{proof}

\begin{theorem}
\begin{equation*}
\E\left[\prod_{t=1}^T \frac{\sum\limits_{n=1}^N w_t^\n}{P_t - 1}\right] = p(y_{1:T}).
\end{equation*}
\end{theorem}

\begin{proof}
Using Lemma \ref{lemma:3} with $t=T, h=T-1$ and
\begin{equation*}
\E\left[\frac{1}{N} \sum_{n=1}^N p\left(y_{1:T} \middle| x_0^\n\right)\right] = p(y_{1:T}).
\end{equation*}
\end{proof}

\section{Biasedness of the marginal likelihood estimator with dynamic thresholds}

Consider the following example: There are two coins on a table, one fair (F) and one biased (B). The probability of getting head (H) and tail (T) using the biased coin are 0.8 and 0.2, respectively. We choose a coin (uniformly at random) and flip it. The probability that the outcome will be head is
\begin{align*}
p(Y=\text{H}) &= p(Y=\text{H}|X=\text{F}) p(X=\text{F}) + p(Y=\text{H}|X=\text{B}) p(X=\text{B}) = \\
&= 0.5 \times 0.5 + 0.8 \times 0.5 = 0.65,
\end{align*}
where $X$ denotes the selected coin (either F or B) and $Y$ denotes the observed outcome (either H or T).

Although this is quite a trivial example, we can still use it to show that using dynamic thresholds, such as quantiles determined in each sweep, may lead to a biased estimate of the marginal likelihood. Note that if the thresholds are constant in all sweeps, the marginal likelihood is unbiased as shown in the proof in Appendix A.

Let us employ a particle filter with rejection control (although there is only one time step) to estimate $p(Y=\text{H})$. We will use the median of the candidate weights after the first propagation of each particle (including the additional particle) as the threshold. To demonstrate the biasedness of the estimator, we consider a filter with only one particle (i.e., $N=1$). The estimator of the marginal likelihood is in this case given by:
\begin{equation*}
    \widehat Z = \frac{w^{(1)}}{P-1}.
\end{equation*}

There exist four possible states of the initial candidate particles (including the additional particle) after the propagation step:
\begin{center}
\begin{tabular}{lccrrrr}
Case & $x'^{(1)}$ & $x'^{(2)}$ & $w'^{(1)}$ & $w'^{(2)}$ & Median \\
\hline
1 & H & H & 0.8 & 0.8 & 0.8 \\
2 & T & T & 0.5 & 0.5 & 0.5 \\
3 & H & T & 0.8 & 0.5 & 0.65 \\
4 & T & H & 0.5 & 0.8 & 0.65 \\
\end{tabular}
\end{center}
Each of these cases is equally likely (the probability of each one being 0.25).

In the first case, both weights are equal to the threshold so both particles are accepted, $w^{(1)}=0.8, P=2$ and $\E[\widehat Z|\text{case 1}] = 0.8$.

Using a similar reasoning we get $\E[\widehat Z|\text{case 2}] = 0.5$ for the second case.

The remaining cases are more difficult. Let $p_F = 0.5/0.65$ denote the acceptance probability of a particle with weight $w'=0.5$, and $p_A = 0.5 \times 1 + 0.5 \times 0.5 / 0.65$ denote the acceptance probability of a restarted particle (its weight being irrelevant).

In the third case, the weight $w^{(1)}$ of the accepted particle is 0.8, but the number $P$ of propagations varies. The expected value of $\widehat Z$ is given by
\begin{align*}
\E[\widehat Z|\text{case 3}] &= p_F \frac{0.8}{2-1} + (1-p_F) \sum_{P=3}^\infty \frac{0.8}{P-1} (1-p_A)^{P-3} p_A \\
&= 0.8 \left(p_F + (1-p_F) \frac{p_A (p_A-\log(p_A)-1)}{(1-p_A)^2}\right) \approx 0.70392
\end{align*}

The fourth case is even more complicated, as we need to distinguish between the case where the weight $w^{(1)}$ of the accepted particle is 0.8 and the case where the weight is 0.65:
\begin{align*}
\E[\widehat Z|\text{case 4}] &= p_F \frac{0.65}{2-1} + (1-p_F) \sum_{P=3}^\infty \frac{0.8}{P-1} (1-p_A)^{P-3} 0.5 + (1-p_F) \sum_{P=3}^\infty \frac{0.65}{P-1} (1-p_A)^{P-3} 0.5 \frac{0.5}{0.65} \\
&= 0.65 \left(p_F + (1-p_F) \frac{p_A-\log(p_A)-1}{(1-p_A)^2}\right) \approx 0.58132
\end{align*}

Finally, we can show that the expected value of $\widehat Z$ is not equal to $p(Y=\text{H})$:
\begin{equation*}
\E[\widehat Z] = 0.25\,\E[\widehat Z|\text{case 1}] + 0.25\,\E[\widehat Z|\text{case 2}] + 0.25\,\E[\widehat Z|\text{case 3}] + 0.25\,\E[\widehat Z|\text{case 4}] \approx 0.64631 \ne 0.65.
\end{equation*}

\end{appendices}

\end{document}